\documentclass[12pt]{article}
\usepackage{amsmath,epsfig}
\usepackage{amssymb}
\usepackage{color}
\usepackage[english]{babel}
\usepackage{graphicx}
\usepackage{times}
\usepackage{verbatim}
\usepackage{multirow}
\usepackage{algorithm}
\usepackage{algorithmic}
\usepackage{lscape}
\usepackage{psfrag}

\makeatletter
\newcommand{\subjclass}[2][1991]{%
  \let\@oldtitle\@title%
  \gdef\@title{\@oldtitle\footnotetext{#1 \emph{Mathematics subject classification.} #2}}%
}
\newcommand{\keywords}[1]{%
  \let\@@oldtitle\@title%
  \gdef\@title{\@@oldtitle\footnotetext{\emph{Key words and phrases.} #1.}}%
}
\makeatother

\newtheorem{thm}{Theorem}[section]
\newtheorem{lem}[thm]{Lemma}

\newtheorem{cor}[thm]{Corollary}
\newtheorem{prop}[thm]{Proposition}

\newenvironment{proof}{\noindent\textsc{Proof: }}{\hfill$\fbox{}$\par\medskip\par}

\newenvironment{aenum}{\begin{enumerate}
 
 }{\end{enumerate}}

\newcommand{\R}{{\mathbb R}}
\newcommand{\Z}{{\mathbb Z}}
\newcommand{\N}{{\mathbb N}}

\newcommand{\cF}{{\cal F}}

\newcommand{\cK}{{\cal K}}

\newcommand{\card}{\mbox{\rm card\,}}

\def\mapright#1{\stackrel{#1}{\longrightarrow}}

\def\mapdown#1{\Big\downarrow\rlap{$\vcenter{\hbox{$\scriptstyle#1$}}$}}

\newcommand{\bdy}{{\partial}} % topological boundary and boundary operator
\newcommand{\cb}{\mbox{\rm \bf cb\,}} % new command: set boundary
\newcommand{\ma}{\texttt{m}\,} % for a matching function from A to B.
\newcommand{\sS}{\texttt{S}} % for sets of cells and canonical bases of chain complexes related to partition.
\newcommand{\sA}{\texttt{A}}
\newcommand{\sB}{\texttt{B}}
\newcommand{\sC}{\texttt{C}} % the set of critial cells
\newcommand{\sL}{\texttt{L}}

\title{A New Matching Algorithm for Multidimensional Persistence}
\author{Madjid Allili, Tomasz Kaczynski\thanks{This work was partially supported by NSERC Canada Discovery Grant  and IMA Minnesota}, Claudia Landi\thanks{Work  performed under the auspices of INdAM-GNSAGA}, Filippo Masoni}


\begin{document}

\maketitle

\begin{abstract}
An algorithm is presented that constructs an acyclic partial matching on the cells of a given simplicial complex from a vector-valued function defined on the vertices and extended to each simplex by taking the least common upper bound of the values on its vertices. The resulting acyclic partial matching may be used to construct a reduced filtered complex with the same multidimensional persistent homology as the original simplicial complex filtered by the sublevel sets of the function. Numerical tests show that in practical cases the rate of reduction in the number of cells achieved by the algorithm is substantial.  This promises to be useful for the computation of multidimensional persistent homology of  simplicial complexes filtered by  sublevel sets of vector-valued functions.
\end{abstract}

\section{Introduction}\label{sec:intro}

In the past decade, Forman's discrete Morse theory \cite{For98,Forman02} appeared to be useful for computing homology of complexes \cite{Har2014} and for providing filtration--preserving reductions of complexes in the study of persistent homology. Until recently, algorithms computing discrete Morse matchings have primarily been used for one--dimensional filtrations, see e.\ g.\ \cite{KinKnuMra05,RobWooShe11,MiNa}. However, there is currently a strong interest in combining persistence information coming from multiple functions in multiscale problems,  e.g. in biological applications  \cite{xia-wei}, which motivates extensions to generalized types of persistence. A more general setting of persistence modules on quiver complexes has been recently presented in \cite{Esc2014}. A parallel attempt in the direction of extending such algorithms to multidimensional filtrations  is our paper \cite{AlKaLa17}. In that paper, an initial framework related to Morse matchings for the multidimensional setting is proposed and an algorithm given by King et al.\ in \cite{KinKnuMra05} is extended. The algorithm produces a partition of the initial complex into three sets $(\sA,\sB,\sC)$ and a bijection $\ma: \sA \to \sB$ called matching. Any simplex which is not matched is added to $\sC$ and declared as critical. The matching algorithm of \cite{AlKaLa17} is used for establishing a reduction of a simplicial complex to a smaller but not necessarily optimal cellular complex. First experiments with filtrations of triangular meshes show that there is a considerable amount of cells identified by the algorithm as critical but which seem to be spurious, in the sense that they appear in clusters of adjacent critical faces which do not seem to carry significant topological information.

The main goal of the current work is to improve the matching method for optimality, in the sense of reducing the number of spurious critical cells. The improvement is in two directions. First, the matching algorithm in \cite{AlKaLa17} is an extension of the one in the 2005 King et al.\ paper \cite{KinKnuMra05} which processes the lower links of vertices and is not optimal even in the one-dimensional setting. Our new matching algorithm extends the one given in 2011 by Robins et al.\ \cite{RobWooShe11} for cubical complexes, which processes lower stars rather than lower links, and improves the result of \cite{KinKnuMra05} for optimality. Next, the new matching algorithm presented here emerges from the observation that, in the multidimensional setting, it is not enough to look at lower stars of vertices: one should take into consideration the lower stars of simplices of all dimensions, as there may be vertices of a simplex which are not comparable in the partial order of the multi-filtration. The vector-valued function initially given on vertices of a complex is first extended to simplices of all dimensions by taking the least common upper bound of the values on their vertices. Then the algorithm processes the lower stars of all simplices, not only the vertices. The resulting acyclic partial matching may be used to construct a reduced filtered Lefschetz complex with the same multidimensional persistent homology as the original simplicial complex filtered by the sublevel sets of the function. This promises to be useful for the computation of multidimensional persistent homology of  simplicial complexes filtered by  sublevel sets of vector-valued functions.

The paper is organized as follows. In Section~\ref{sec:prel}, the preliminaries are introduced. We recall the definition of simplicial complex, which is the input for our matching algorithm, of partial and acyclic matching, and of multidimensional filtration. A preliminary topological sorting Algorithm~\ref{alg:sorting} is presented.

In Section~\ref{sec:alg}, the main Algorithm~\ref{alg-match} is presented. Its correctness is proved and the complexity analyzed.

Section~\ref{per-hom} starts from recalling the notion of Lefschetz complex \cite{Lef42}, also studied under the name of $\sS$--complex in \cite{MrBa09}. These complexes are produced by applying the reduction method \cite{KaMrSl98,MrBa09,MiNa} to an initial simplicial complex, with the use of the matchings produced by our main Algorithm~\ref{alg-match}. Multidimensional persistent homology and the reduction method are recalled. The main feature of the reduction method is preserving persistent homology. It is presented in Corollary~\ref{cor:homology-S-iso-C}.

In Section~\ref{sec:experiments} experiments on synthetic and real 3D data are presented. These experiments show that in practical cases the rate of reduction in the number of cells achieved by the algorithm is substantial. Statistics are presented in Tables \ref{tab:sphere}, \ref{tab:torus}, \ref{tab:klein}, and \ref{tab:space}.

The improvement is observed in practice but the optimality of Morse reduction is not yet well defined in the multidimesional setting. Recall that, in the classical smooth case, the singularities of vector-valued functions on manifolds are in general not isolated; they form submanifolds of lower dimension. An appropriate extension of the Morse theory to multidimensional functions is not much investigated yet. Some related work is that of Edelsbrunner and Harer \cite{EdHa02} on Jacobi sets and of Patel \cite{Pat10} on preimages of maps between manifolds. However there are essential differences between those concepts and our sublevel sets with respect to the partial order relation.

\section{Working Assumptions}\label{sec:prel}

\subsection{Simplicial framework}\label{sec:simpl}

In this paper, we shall primarily work in the framework of finite geometric simplicial complexes $\cK$ in the Euclidean space $\R^d$. More precisely, the main result of this paper, the Matching Algorithm~\ref{alg-match} takes as input any finite geometric simplicial complex $\cK$ in the Euclidean space $\R^d$. On the other hand, its applications to computing multidimensional persistent homology of $\cK$ are expressed in the language of Lefschetz complexes which are an algebraic abstraction of cellular complexes. The structure of Lefschetz complexes and the notion of persistent homology will be reviewed in Section~\ref{per-hom}.

Let us recall that a {\em  geometric $q$-simplex} $\sigma = [v_0,v_1, \ldots, v_q]\in\cK$ is the convex hull of $q + 1$ affinely independent vertices $v_0$,$v_1$, $\ldots$, $v_q$ in $\R^d$. The number $q$ is its dimension. In programming, $\sigma$ can be identified with the list of its vertices, which is called an {\em abstract simplex}. A {\em face} of $\sigma$ is a simplex $\tau$ whose vertices constitute a subset of $\{v_0,v_1,\ldots ,v_q\}$. If $\dim \tau = q-1$, it is called a {\em primary face} or, for short, a {\em facet} of $\sigma$. In this case, $\sigma$ is called a {\em primary coface} or a {\em cofacet} of $\tau$, and we write $\tau < \sigma$.

The set $\cK_q\subset \cK$ is a collection of all simplices in $\cK$ dimension $q$. In particular, $\cK_0$ is the set of vertices in $\cK$. Given $\sigma\in\cK$, we write $\cK_0(\sigma)$ for the set of vertices of $\sigma$. The collection $\cK$ is called a {\em simplicial complex} if every face of a simplex in $\cK$ is also in $\cK$ and the intersection of any two simplices in $\cK$, if non-empty, is their common face. We denote either by $K$ or by $|\cK|$ the {\em carrier} of  $\cK$ which is the union of all its simplices. Thus $K$ is a polyhedron in $\R^d$ and $\cK$ its triangulation.

\subsection{Partial matching}\label{sec:match}

A {\em partial matching} $(\sA,\sB,\sC,\ma)$ on $\cK$ is a partition of $\cK$ into three
sets $\sA,\sB,\sC$ together with a bijective  map $\ma: \sA\to \sB$  such that, for each $\tau\in \sA$, $\ma(\tau)$ is a cofacet of $\tau$.

An {\em $\ma$--path} is a sequence
\begin{equation}\label{eq:m-path}
\sigma_0,\tau_0,\sigma_1,\tau_1,\ldots ,\sigma_p,\tau_p,\sigma_{p+1}
\end{equation}
such that, for each $i=0,\ldots ,p$, $\sigma_{i+1}\ne \sigma_i$, $\tau_i=\ma(\sigma_i)$, and $\tau_i$ is a
cofacet of  $\sigma_{i+1}$.

A partial matching $(\sA,\sB,\sC,\ma)$ on $\cK$ is called {\em acyclic} if there does not exist a closed $\ma$--path, that is a path as in (\ref{eq:m-path}) such that, $\sigma_{p+1}=\sigma_0$.

A convenient way to phrase the definition of an acyclic partial matching is via the {\em Hasse diagram} of $\cK$. It is the directed graph whose vertices are elements of $\cK$, edges are given by cofacet relations, and oriented from the larger element to the smaller one. Given a partial matching $(\sA,\sB,\sC,\ma)$ on  $\cK$, we change the orientation of the edge $(\tau, \sigma)$ whenever $\tau = \ma(\sigma)$. Thus the $\ma$--path in (\ref{eq:m-path}) can be displayed as
\begin{equation}\label{eq:Hasse}
\sigma_0 \xrightarrow{\ma}  \tau_0 \xrightarrow{>} \sigma_1 \xrightarrow{\ma} \tau_1 \xrightarrow{>} \ldots \xrightarrow{>} \sigma_n \xrightarrow{\ma}
\tau_p \xrightarrow{>} \sigma_{p+1}
\end{equation}
where $\ma$ stands for the matching, and the symbol $>$ for the cofacet relation. The acyclicity means that the oriented graph obtained this way, which is also called the {\em modified Hasse diagram} of $\cK$, has no nontrivial cycles. A directed graph with no directed cycles is called a directed acyclic graph (DAG). Thus, a partial matching $(\sA,\sB,\sC,\ma)$ on $\cK$ is acyclic if its corresponding modified Hasse diagram is a DAG.

\medskip

It is important to emphasize that the sets of simplices $\sA, \sB, \sC$  of $\cK$ in general are not simplicial complexes. Nevertheless, the set $\sC$ of critical simplices can be given a combinatorial structure of a Lefschetz complex discussed in Section~\ref{per-hom} and a  topology of a CW complex \cite{Mun84}.

%\section{Working assumptions}\label{sec:hp}
\subsection{Multidimensional filtration}\label{sec:md-f}

The main goal of this paper is to produce a partial matching which preserves the filtration of $\cK$ by sublevel sets of a vector-valued function $f: \cK_0 \to \R^k$ given on the set of vertices of $\cK$. We define these notions below.

\medskip

Let now $\cK$ be a  simplicial complex of cardinality $N$, and let $n$ denote the cardinality of $\cK_0$.

We assume that $f: \cK_0 \rightarrow \mathbb{R}^k$ is a function  on the set of vertices of $\cK$ which is {\em component-wise injective},
that is, whose components $f_i$ are injective. Note that  this is a stronger condition than assuming that $f$ is injective.

Given any function $\tilde f: \cK_0 \rightarrow \mathbb{R}^k$, we can obtain a component-wise injective function $f$ which is arbitrarily
close to $\tilde f$ via the following procedure.
For $i=1,\ldots ,k$, let us set $\eta_i=\min\{|\tilde f_i(v)-\tilde f_i(w)|: v,w\in \cK_0 \wedge \tilde f_i(v)\ne \tilde f_i(w)\}$.
For each $i$ with $1\le i\le k$, we can assume that the $n$ vertices in $\cK_0$ are indexed by a integer index $j$, with $1\le j\le n$, increasing with $\tilde f_i$. Thus, the
function $f_i:\cK_0 \rightarrow \mathbb{R}$  can be  defined by setting $f_i(v_j)=\tilde f_i(v_j)+j\eta_i/ n^s$, with $s\ge 1$
(the larger $s$, the closer $f$ to $\tilde f$). Finally, it is sufficient to set $f=(f_1,f_2,\ldots, f_k)$.

We extend $f$ to a function $f: \cK \to \R^k$ as follows:
\begin{eqnarray}\label{max}
f(\sigma) = (f_1 (\sigma), \ldots, f_k (\sigma)) \qquad \mbox{with} \qquad f_i (\sigma) = \max_{v \in \cK_0 (\sigma)} f_i (v).
\end{eqnarray}

Any function $f : {\cK} \to \R^k$  that is an extension of a component-wise injective function $f : {\cK}_0 \to \R^k$ defined on
the vertices of the complex $\cK$ in such a way that $f$ satisfies equation~(\ref{max}) will be called \emph{admissible}.

\medskip

In $\R^k$ we consider the following partial order. Given two values $a=(a_i), b=(b_i)\in\R^k$ we set $a\preceq b$ if and only if
$a_i\le b_i$ for every $i$ with $1\le i\le k$. Moreover we write $a\precneqq b$ whenever $a\preceq b$ and $a\ne b$.

The {\em sublevel set filtration} of $\cK$ induced by an admissible function $f$ is the family $\{\cK^a\}_{a\in \R^k}$ of subsets of $\cK$ defined as follows:
\[
{\cK}^a=  \{\sigma=[v_0,v_1,\ldots,v_q]\in {\cK} \mid f(v_i)\preceq a, \ i=0,\ldots ,q\}.
\]
It is clear that, for any parameter value $a\in\R^k$ and any simplex $\sigma\in {\cK}^a$, all faces of $\sigma$ are also in ${\cK}^a$. Thus ${\cK}^a$ is a simplical subcomplex of $\cK$ for each $a$. The changes of topology of ${\cK}^a$ as we change the multiparameter $a$ permit recognizing some features of the shape of  $|{\cK}|$ if $f$ is appropriately chosen. For this reason, the function $f$ is called in the literature a {\em measuring function} or, more specifically, a {\em multidimensional measuring function} \cite{BiCe*08}.

\medskip

The {\em lower star} of a simplex is the set
\[
L(\sigma) = \{ \alpha \in \cK \mid  \sigma \subseteq \alpha
\quad \mbox{and} \quad f(\alpha) \preceq f(\sigma)\},
\]
and the {\em reduced lower stars} is the set
$L_*(\sigma)=L(\sigma)\setminus \{\sigma\}$.

\subsection{Indexing map}

An {\em indexing map} on the simplices of the complex $\cK$, compatible with an admissible function $f$, is a bijective map $I:{\cK} \to \{1,2,\ldots, N\}$ such that, for each
$\sigma, \tau \in {\cK}$ with $\sigma \ne \tau$, if $\sigma\subseteq \tau$ or $f(\sigma)\precneqq  f(\tau)$ then $I(\sigma)<I(\tau)$.

To build an indexing map $I$ on the simplices of the complex $\cK$, we will revisit the algorithm introduced in~\cite{AlKaLa17} that uses the topological sorting of a DAG
to build an indexing for vertices of a complex that is compatible with the ordering of values of a given function defined on the vertices.
We will extend the algorithm to build an indexing for all cells of a complex that is compatible with both the ordering of values of a given admissible
function defined on the cells and the ordering of the dimensions of the cells.

We recall that a topological sorting of a directed graph is a linear ordering of its nodes such that for every directed edge
$(u, v)$ from node $u$ to node $v$, $u$ precedes $v$ in the ordering. This ordering is possible if and only if the graph
has no directed cycles, that is, if it is a  DAG.

A simple well known algorithm (see~\cite{Wikipedia14,Kahn62}) for this task consists
of successively finding nodes of the DAG that have no incoming edges and placing them in a list for the final sorting. Note that
at least one such node must exist in a DAG, otherwise the graph must have at least one directed cycle.
The algorithm consists of two nested loops as presented in Algorithm \ref{alg:sorting}.

\begin{algorithm}
\caption{Topological sorting}
\label{alg:sorting}
\begin{algorithmic}[1]
\STATE {\bf Input:} A DAG whose list of nodes with no incoming edges is  \texttt{I}
\STATE {\bf Output:} The list \texttt{L} containing the sorted nodes
\WHILE {there are nodes remaining in \texttt{I}}
\STATE remove a node $u$ from \texttt{I}
\STATE add $u$ to \texttt{L}
\FOR {each node $v$ with an edge $e$ from $u$ to $v$ }
\STATE remove edge $e$ from the DAG
\IF {$v$ has no other incoming edges}
\STATE insert $v$ into \texttt{I}
\ENDIF
\ENDFOR
\ENDWHILE
\end{algorithmic}
\end{algorithm}

When the graph is a DAG, there exists at least one solution for the sorting problem, which is not necessarily unique.
We can easily see that each node and each edge of the DAG is visited once by the algorithm, therefore its running time
is linear in the number of nodes plus the number of edges in the DAG.

\begin{lem}\label{lem:index}
Let $f:\cK\to\R^k$ be an admissible function. There exists an injective function $I:{\cK} \to \N$ such that, for each $\sigma, \tau \in {\cK}$ with $\sigma \ne \tau$,
if $\sigma\subseteq \tau$ or $f(\sigma)\precneqq  f(\tau)$ then $I(\sigma)<I(\tau)$.
\end{lem}

\begin{proof}
The set ${\cK}$ is partially ordered by the following relation: $\sigma \sqsubseteq \tau$ if and only if either
$\sigma=\tau$ or $\sigma\ne \tau$ and, in the latter case, $\sigma$ is a face of $\tau$ or  $f(\sigma)\precneqq  f(\tau)$.
Indeed, it can be straightforwardly checked that this relation is reflexive, antisymmetric and transitive.
Hence $({\cK}, \sqsubseteq)$ can be represented in a directed graph by its Hasse diagram that is acyclic.

The topological sorting Algorithm \ref{alg:sorting}  allows us to sort and store the simplices
in $\cK$ in an array $\sL$ of size $N$, with indexes that can be chosen from 1 to $N$. It follows that the map $I:{\cK} \to \{1,2,\ldots,N\}$ that
associates to every node its index in the array $\sL$ is bijective. Moreover, and due to the topological sorting, $I$ satisfies the constraint
that for $\sigma, \tau \in {\cK}$ with $\sigma \ne \tau$, if $\sigma\subseteq \tau$ or $f(\sigma)\precneqq  f(\tau)$, then $I(\sigma)<I(\tau)$.
\end{proof}

\section{Matching Algorithm}
\label{sec:alg}

\subsection{The algorithm}

The main contribution of this paper is the Matching Algorithm~\ref{alg-match}.
It uses the following input data:

\begin{enumerate}
\item A finite simplicial complex $\cK$.

\item An admissible function $f:\cK\to \R^k$. Typically, it is obtained from a component-wise injective function $f : {\cK}_0 \to \R^k$   using  the extension formula given in equation~(\ref{max}).

\item An indexing map $I$ compatible with $f$. It can be precomputed using the topological sorting Algorithm \ref{alg:sorting}.
\end{enumerate}

We also use the following definitions:

\begin{enumerate}\setcounter{enumi}{3}
\item Given a simplex $\sigma$, we use \texttt{{unclass}\_{facets}}$_{\sigma}(\alpha$)
to denote the set of facets of a simplex $\alpha$ that are
in $L(\sigma)$ and have not been classified yet, that is, not inserted
in either $\sA$, $\sB$, or $\sC$, and \texttt{{num}\_{unclass}\_{facets}}$_{\sigma}(\alpha$) to denote the cardinality of
\texttt{{unclass}\_{facets}}$_{\sigma}(\alpha$).

\item We initialize \texttt{classified}($\sigma$)={\bf false}
for every $\sigma\in \cK$.

\item We use priority queues \texttt{PQzero} and \texttt{PQone}
which store candidates for pairings with zero and one unclassified
facets respectively in the order given by $I$. We initialize both as empty sets.
\end{enumerate}

The algorithm processes cells in the increasing order of their indexes given by the indexing map $I$ defined on $\cK$. Each cell $\sigma$ can be set to the states
of \texttt{classified}($\sigma$)={\bf true} or \texttt{classified}($\sigma$)={\bf false} so that if it is processed as part of a lower star of another cell it is not
processed again by the algorithm. The algorithm makes use of extra routines to calculate the cells in the lower star $L(\sigma)$ and the set of unclassified
faces \texttt{{unclass}\_{facets}}$_{\sigma}(\alpha$) of $\alpha$ in $L_*(\sigma)$ for each cell $\sigma \in \cK$ and each cell $\alpha \in L_*(\sigma)$.
The goal of the processing is to build a partition of $\cK$ into three lists $\sA$, $\sB$, and $\sC$ where $\sC$ is the list
of critical cells and in which each cell in $\sA$ is paired in a one-to-one manner with a cell in $\sB$ which defines a bijective map $\ma:\sA\to \sB$.
When a cell $\sigma$ is considered, each cell in its lower star $L(\sigma)$ is processed exactly once as shown in Lemma~\ref{lem:correctness3}. The cell $\sigma$ is inserted into the list of critical cells $\sC$ if $L_* (\sigma) = \emptyset$. Otherwise,
$\sigma$ is paired with the cofacet $\delta \in L_* (\sigma)$ that has minimal index value $I(\delta)$. The algorithm makes additional pairings which can be
interpreted topologically as the process of
constructing $L_*(\sigma)$ with simple homotopy expansions or the process of reducing $L_*(\sigma)$ with simple homotopy contractions.
When no pairing is possible a  cell is classified as critical and the
process is continued from that cell. A cell $\alpha$ is candidate for a pairing when \texttt{{unclass}\_{facets}}$_{\sigma}(\alpha$) contains exactly one element $\lambda$ that belongs to \texttt{PQzero} as shown in Lemma~\ref{lem:correctness1}. For this purpose, the priority queues \texttt{PQzero} and \texttt{PQone} which store cells with zero and
one available unclassified faces respectively are created. As long as \texttt{PQone} is not empty, its front is popped and either inserted into \texttt{PQzero} or paired with its single available unclassified face. When \texttt{PQone} becomes empty, the front cell of \texttt{PQzero} is declared as critical and inserted in $\sC$.

\begin{algorithm}[H]
\caption{Matching}
\label{alg-match}
\begin{algorithmic}[1]
\STATE {\bf Input:} A finite simplicial complex $\cK$ with an admissible function $f : {\cK} \to \R^k$
and an indexing map $I:{\cK} \to \{1,2,\ldots, N\}$ on its simplices compatible with $f$.
\STATE {\bf Output:} Three lists $\sA,\sB,\sC$ of simplices of $\cK$, and a function $\ma:\sA\to \sB$.\\
\FOR{$i=1$ to $N$}
\STATE $\sigma:=I^{-1}(i)$
\IF{\texttt{classified}($\sigma$)=\FALSE}
\IF{$L_*(\sigma)$ contains no  cells}
\STATE add $\sigma$ to $\sC$, \texttt{classified}($\sigma$)=\TRUE
\ELSE
\STATE $\delta:=$ the cofacet in $L_*(\sigma)$ of minimal index $I(\delta)$
\STATE add $\sigma$ to $\sA$ and $\delta$ to $\sB$ and define $\ma (\sigma) = \delta$, \texttt{classified}($\sigma$)=\TRUE, \texttt{classified}($\delta$)=\TRUE
\STATE add all  $\alpha \in L_*(\sigma)-\{\delta\}$ with \texttt{{num}\_{unclass}\_{facets}}$_{\sigma} (\alpha)= 0$ to \texttt{PQzero}
\STATE add all  $\alpha \in L_*(\sigma)$ with \texttt{{num}\_{unclass}\_{facets}}$_{\sigma} (\alpha)$ = 1 and  $\alpha > \delta$ to \texttt{PQone}
\WHILE{\texttt{PQone} $\neq \emptyset$ or \texttt{PQzero} $\neq \emptyset$}
\WHILE{\texttt{PQone} $\neq \emptyset$}
\STATE $\alpha :=$ \texttt{PQone}.pop\_front
\IF{\texttt{{num}\_{unclass}\_{facets}}$_{\sigma} (\alpha$) = 0}
\STATE add $\alpha$ to \texttt{PQzero}
\ELSE
\STATE add $\lambda \in \texttt{{unclass}\_{facets}}_{\sigma} (\alpha)$ to $\sA$, add $\alpha$ to $\sB$ and define $\ma (\lambda) = \alpha$,
\texttt{classified}($\alpha$)=\TRUE, \texttt{classified}($\lambda$)=\TRUE
\STATE remove $\lambda$ from \texttt{PQzero}
\STATE add all  $\beta \in L_*(\sigma)$ with \texttt{{num}\_{unclass}\_{facets}}$_{\sigma} (\beta$) = 1 and
either $\beta > \alpha$ or $\beta > \lambda$ to \texttt{PQone}
\ENDIF
\ENDWHILE
\IF{\texttt{PQzero} $\neq \emptyset$}
\STATE $\gamma :=$ \texttt{PQzero}.pop\_front
\STATE add $\gamma$ to $\sC$, \texttt{classified}($\gamma$)=\TRUE
\STATE add all  $\tau \in L_*(\sigma)$ with \texttt{{num}\_{unclass}\_{facets}}$_{\sigma} (\tau$) = 1 and
$\tau > \gamma$ to \texttt{PQone}
\ENDIF
\ENDWHILE
\ENDIF
\ENDIF
\ENDFOR
\end{algorithmic}
\end{algorithm}

%%%%%%%%%%%%%%%%%%%%%%%%%%%%%%%%%%%%%%%%%%%%%%%%%%%%%%%%%%%%%%%%%
%%%%%%%%%%%%%%%%%%%%%%%%%%%%%%%%%%%%%%%%%%%%%%%%%%%%%%%%%%%%%%%%%
%%%%%%%%%%%%%%%%%%%%%%%%%%%%%%%%%%%%%%%%%%%%%%%%%%%%%%%%%%%%%%%%%

We illustrate the algorithm by a simple example. We use the simplicial complex $\sS$ from our first paper \cite[Figure 2]{AlKaLa17} to compare the outputs of the previous matching algorithm and the new one. Figure~\ref{fig:matching2}(a) displays $\sS$ and the output of \cite[Algorithm 6]{AlKaLa17}.

The coordinates of vertices are the values of the function considered in \cite{AlKaLa17}. Since that function is not component-wise injective, we denote it by $\tilde{f}$ and we start from constructing a component-wise injective approximation $f$ discussed at the beginning of Section~\ref{sec:md-f}. We have $\eta_i=1$, $i=1,2$. Let $\epsilon=\epsilon_i=5^{-s}$ where $s\geq 1$ can be conveniently fixed for programming. For increasing order of $f_1$, the vertices are ordered as $(v_0,v_1,v_2,v_3,v_4)$. We get the values
\[
f_1(v_0)=\epsilon, f_1(v_1)=1+2\epsilon, f_1(v_2)=1+3\epsilon, f_1(v_3)=2+4\epsilon, f_1(v_4)=2+5\epsilon.
\]
For increasing order of $f_2$, the vertices are ordered as $(v_0,v_1,v_3,v_2,v_4)$. We get the values
\[
f_2(v_0)=\epsilon, f_2(v_1)=1+2\epsilon, f_2(v_3)=3\epsilon, f_2(v_2)=1+4\epsilon, f_2(v_4)=1+5\epsilon.
\]
If we interpret the passage from $\tilde{f}$ to $f$ as a displacement of the coordinates of vertices, the new complex $\cK$ is illustrated by Figure~\ref{fig:matching2}(b). The partial order relation is preserved when passing from $\tilde{f}$ to $f$, and the indexing of vertices in \cite[Figure 2]{AlKaLa17} may be kept for $f$. Hence, it is easy to see that \cite[Algorithm 6]{AlKaLa17} applied to $\cK$ gives the same result as that displayed in Figure~\ref{fig:matching2}(a).

In order to apply our new Algorithm~\ref{alg-match}, we need to index all $14$ simplices of $\cK$. For convenience of presentation, we label the vertices $w_i$, edges $e_i$, and triangles $t_i$ by the index values $i=1,2,...,14$. The result is displayed in Figure~\ref{fig:matching2}(b). The sequence of vertices $(v_0,v_1,v_2,v_3,v_4)$ is replaced by $(w_1,w_2,w_4,w_8,w_{12})$.  Here are the main steps of the algorithm:

\begin{center}
\begin{tabular}[b]{ll}
$i=1$: &   $L_*(w_1)=\emptyset$, add $w_1$ to $\sC$.\\
$i=2$: & $L_*(w_2)=\{e_3\}$, $\ma(w_2)=e_3$.\\
$i=3$: &  $e_3$ classified.\\
$i=4$: & \parbox[t]{10cm}{$L_*(w_4)=\{e_5,e_6,t_7\}$, $\ma(w_4)=e_5$,\\
 add $e_6$ to $\texttt{PQzero}$, add $t_7$ to $\texttt{PQone}$,\\
at line $15$, remove $\alpha=t_7$ from $\texttt{PQone}$,\\
 at line $19$, $\lambda=e_6$, $\ma(e_6)=t_7$, remove $e_6$ from $\texttt{PQzero}$.}\\
$i=5,6,7$: & $e_5, e_6, t_7$ classified.\\
 $i=8$: &  \parbox[t]{10cm}{ $L_*(w_8)=\{e_9,e_{10},t_{11}\}$, $\ma(w_8)=e_9$, \\
		add $e_{10}$ to $\texttt{PQzero}$, add $t_{11}$ to $\texttt{PQone}$,\\
    at line $15$, remove $\alpha=t_{11}$ from $\texttt{PQone}$, \\
    at line $19$, $\lambda=e_{10}$, $\ma(e_{10})=t_{11}$, remove $e_{10}$ from $\texttt{PQzero}$.}\\
$i=9,10,11$: & $e_9, e_{10}, t_{11}$ classified.\\
$i=12$: & \parbox[t]{10cm}{ $L_*(w_{12})=\{e_{13},e_{14}\}$, $\ma(w_{12})=e_{13}$, \\
		add $e_{14}$ to $\texttt{PQzero}$, $\texttt{PQone}=\emptyset$,\\
    at line $25$, add $\gamma=e_{14}$ to $\sC$.}\\
$i=13,14$: &  $e_{13}, e_{14}$ classified.
\end{tabular}
\end{center}

The output is displayed in Figure~\ref{fig:matching2}(b).

\begin{figure}[h]
\begin{center}
\begin{tabular}{ccc}
\psfrag{v0}{$v_0=(0,0)$}
\psfrag{v1}{$v_1=(1,0)$}
\psfrag{v2}{$v_2=(1,1)$}
\psfrag{v3}{$v_3=(2,0)$}
\psfrag{v4}{$v_4=(2,1)$}
 \includegraphics[width=0.40\textwidth]{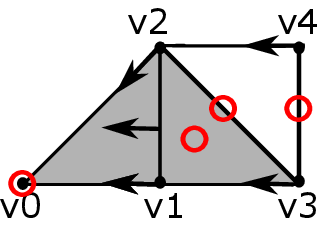}
& \ \ \ \ \  \ \ \ \ \  &
\psfrag{w1}{$w_1$}
\psfrag{w2}{$w_2$}
\psfrag{w4}{$w_4$}
\psfrag{w12}{$w_{12}$}
\psfrag{w8}{$w_8$}
\psfrag{e3}{$e_3$}
\psfrag{e5}{$e_5$}
\psfrag{e6}{$e_6$}
\psfrag{e9}{$e_9$}
\psfrag{e10}{$e_{10}$}
\psfrag{e13}{$e_{13}$}
\psfrag{e14}{$e_{14}$}
\psfrag{t7}{$t_7$}
\psfrag{t11}{$t_{11}$}
\includegraphics[width=0.45\textwidth]{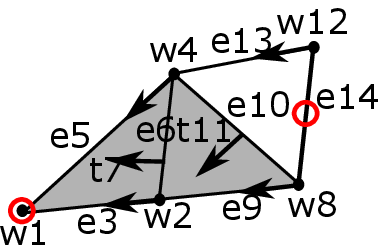}
\\
(a) &\ \ \ \ \  & (b)
\end{tabular}
\caption{In (a), the complex and output of Algorithm 6 of \cite{AlKaLa17} are displayed. Gray-shaded triangles are those which are present in the simplicial complex. Critical simplexes are marked by red circles and the matched simplexes are marked by arrows. In  (b), the complex is modified so to satisfy the coordinate-wise injectivity assumption. Labeling of all simplices by the indexing function and the output of Algorithm~\ref{alg-match} are displayed.
}
\label{fig:matching2}
\end{center}
\end{figure}

%%%%%%%%%%%%%%%%%%%%%%%%%%%%%%%%%%%%%%%%%%%%%%%%%%%%%%%%%%%%%%%%%
%%%%%%%%%%%%%%%%%%%%%%%%%%%%%%%%%%%%%%%%%%%%%%%%%%%%%%%%%%%%%%%%%
%%%%%%%%%%%%%%%%%%%%%%%%%%%%%%%%%%%%%%%%%%%%%%%%%%%%%%%%%%%%%%%%%

\subsection{Correctness}

Recall that $f=(f_1,\ldots, f_k): \cK_0 \rightarrow \mathbb{R}^k$ is a function  whose  components $f_i$ are injective on the vertices of $\cK$; moreover, $f$ is  extended to $f=(f_1,\ldots, f_k): \cK \rightarrow \mathbb{R}^k$ defined on cells $\sigma$ of any dimension by using formula (\ref{max}).  The assumption that $f$  is component-wise injective on the vertices is not sufficient to obtain disjoint lower stars, but when two lower
stars meet, then they get classified at the same time. This is expressed by the following statements.

\begin{lem}
\label{lem:inj}
The following statements hold:
\begin{enumerate}
\item[(1)] If $\tau\in L(\sigma)$, then $f(\tau)=f(\sigma)$.
\item[(2)] If $\tau\in L_*(\sigma)$, then $I(\sigma)<I(\tau)$.
\item[(3)] If $f(\sigma)=f(\tau)$ then there exists $\alpha \subseteq \sigma\cap \tau$ with $f(\alpha)=f(\sigma)=f(\tau)$.
\item[(4)] Assume that $\sigma_1$ and $\sigma_2$ are two  simplices  of $\cK$ such that $L(\sigma_1) \cap L(\sigma_2) \neq \emptyset$.
Then, there exists a simplex $\beta \in \cK$ such that $L(\sigma_1) \cup L(\sigma_2) \subseteq L(\beta)$ and $I(\beta) \le \min \{ I(\sigma_1), I(\sigma_2)\}$.
\end{enumerate}
\end{lem}

\begin{proof}
$(1)$ If $\tau\in L(\sigma)$, then $f(\tau)\preceq f(\sigma)$ by definition of lower star. On the other hand,  since
$\sigma\subseteq \tau$, by definition of $f$, $f(\sigma)\preceq f(\tau)$. Thus $f(\sigma)=f(\tau)$.

\noindent
$(2)$ If $\tau\in L_*(\sigma)$, then $\sigma\subset \tau$ and the conclusion follows from the definition of the indexing map.

\noindent
$(3)$ If $f(\sigma)=f(\tau)$, then, for every $i$, $\max_{v\in\cK_0(\sigma)}f_i(v)=\max_{v\in\cK_0(\tau)}f_i(v)$. By
the injectivity of $f_i$, the two maxima must be attained at the same vertex. Therefore $\sigma$ and $\tau$ have a common face.

\noindent
$(4)$  If  there exists a simplex
$\gamma \in L(\sigma_1) \cap L(\sigma_2)$, then  we get $f(\gamma) = f(\sigma_1) = f(\sigma_2)$ from $(1)$. By  $(3)$,
there exists a simplex $\beta \subseteq \sigma_1 \cap \sigma_2$ such that $f(\beta) = f(\sigma_1) = f(\sigma_2)$. It is now clear that for any
$\delta \in L(\sigma_1) \cup L(\sigma_2)$, $\beta \subseteq \delta$ and  $f(\delta) = f(\sigma_1) = f(\sigma_2) = f(\beta)$, thus
$\delta \in L(\beta)$. By $(2)$, $I(\beta) \leq \min \{ I(\sigma_1), I(\sigma_2)\}$.
\end{proof}

\begin{lem}\label{lem:correctness1}
Assume that $\sigma$ is a cell in $\cK$. If $\alpha \in L_* (\sigma)$ is a cofacet of $\sigma$ then, at any stage of the algorithm , $\texttt{{num}\_{unclass}\_{facets}}_{\sigma} (\alpha)\le 1$, and it is equal to 1 if and only if $\sigma$ is still unclassified. In this case, the unclassified face of $\alpha$ is exactly $\sigma$.
\end{lem}

\begin{proof}
Let us assume that $\texttt{{num}\_{unclass}\_{facets}}_{\sigma} (\alpha)\ge 1$. For any unclassified face  $\gamma$ of $\alpha$ such that $\gamma \in L(\sigma)$, it holds that
$\dim \gamma = \dim \sigma$. Indeed, $\dim \gamma<\dim \alpha=\dim\sigma+1$, and $\dim\gamma\ge \dim\sigma$ because $\gamma\in L(\sigma)$. Thus, if $\gamma\ne \sigma$, the assumption $\gamma \in L(\sigma)$ is contradicted. Hence, if $\gamma\ne \sigma$ for every $\gamma$, we have $\texttt{{num}\_{unclass}\_{facets}}_{\sigma} (\alpha)=0$.  Otherwise, if $\gamma=\sigma$, then $\texttt{{num}\_{unclass}\_{facets}}_{\sigma} (\alpha)= 1$, concluding the proof.
\end{proof}

\begin{lem}\label{lem:i=1-2}
Let $\sigma\in \cK$ such that  $I(\sigma)=i$ with $i\le 2$. Then, $\sigma$ is a vertex. Moreover,
\begin{enumerate}
\item if $i=1$, then Algorithm~\ref{alg-match} classifies $\sigma$ as critical at line 7;
\item if $i=2$, then either $L(\sigma) = \{ \sigma\}$ or $L(\sigma) = \{ \sigma, \delta\}$ where $\delta$ is an edge whose vertices are the cells with indexes 1 and 2. Moreover, if $L(\sigma) = \{ \sigma\}$, then $\sigma$ is classified as critical at line 7; if $L(\sigma) = \{ \sigma, \delta\}$, then $\delta$ is an edge and $\sigma$ is  paired with  $\delta$ at line 9.
\end{enumerate}
\end{lem}

\begin{proof}
If $I(\sigma)\le 2$, then $\sigma$ needs to be a vertex. Indeed, if we assume  $\dim \sigma \geq 1$, then it can be written as $\sigma = \langle v_0, v_1, \ldots, v_k \rangle$ with $k \geq 1$. It follows that $\sigma$ has at least two faces of lower dimension and lower value by $f$ which should also have lower indexes than that of $\sigma$.
This contradicts the fact that $I(\sigma) \le 2$.

Let us now prove separately statements $1$ and $2$.

1. We note that \texttt{classified}($\sigma$)={\bf false} at line 5 because of the initialization. Moreover, $L_*(\sigma)$ is empty. To see this, let us observe that, for any coface  $\gamma$ of $\sigma$,  it must hold that $f_i(\sigma)<f_i(\gamma)$ for at least one index $i=1,\ldots, k$. Indeed, $\sigma$ is a vertex of $\gamma$ and at any other vertex of $\gamma$   the value of $f_i$ must be greater than $f_i(\sigma)$ because $f_i$ is injective and $\sigma$ has minimal index. Hence $\sigma$ gets classified at line 7 and there is no other cell in $L(\sigma)$ to classify.

2.  If $\alpha \in L_*(\sigma)$, then all the vertices in $\alpha$ other than $\sigma$ should have  $f$ values lower than $\sigma$.
They should therefore have lower indexes too. The only possibility left is to have $\alpha = \langle v, w \rangle$ where $I(v) = 1$ and $I(w)=2$.
\end{proof}

\begin{lem}\label{lem:line10}
  Assume that $\sigma\in \cK$ is unclassified when Algorithm~\ref{alg-match} reaches line 5 for $i=I(\sigma)$, and that, for all simplexes $\beta\in \cK$ with $I(\beta)<I(\sigma)$, it holds that \texttt{classified}$(\gamma)$={\bf true} for all cells $\gamma\in L(\beta)$. Then the following statements hold true:
\begin{enumerate}
\item[(i)] All simplexes in $L(\sigma)$ are  also unclassified at step 5 when $i=I(\sigma)$;
\item[(ii)] If Algorithm~\ref{alg-match} gets to line 9, then there exists at least one cofacet of $\sigma$. Moreover, the one with minimal index, say $\delta$, has exactly $\sigma$ as unclassified facet, and it is still unclassified. Thus $\sigma$ and $\delta$ get classified at line 10.
\item[(iii)] If $\alpha\in L_*(\sigma)$ and $\texttt{{num}\_{unclass}\_{facets}}_{\sigma} (\alpha)=0$ at line 11 of Algorithm~\ref{alg-match}, then $\alpha$ is a facet of $\sigma$.
\end{enumerate}
\end{lem}

\begin{proof}
\begin{enumerate}
\item[{\em (i)}] If $L(\sigma)=\{\sigma\}$ the claim is true by assumption. Let us assume that  $\sigma$ has at least one coface $\alpha \in L_*(\sigma)$.  If $\alpha$ is  classified  it  belongs to the lower star of another cell $\beta$ different from $\sigma$ with $I(\beta)<I(\sigma)$. By Lemma~\ref{lem:inj}(4), also $\sigma$ belongs to $L(\beta)$ and, therefore, $\sigma$ is already classified by assumption. This gives a contradiction. Hence $\alpha$ is not classified.

\item[{\em (ii)}] If $\sigma$ had no cofaces, then $L_*(\sigma)$ would be empty. Therefore line 9 would not be reached, contradicting the hypothesis. So $\sigma$ has at least one coface $\alpha \in L_*(\sigma)$. By Lemma~\ref{lem:inj}(1), $f(\alpha)=f(\sigma)$. Assuming $\dim \sigma = p$ and $\dim \alpha = p+r$, there exists a sequence of simplices $\alpha_1, \ldots \alpha_{r-1}$ of dimensions $p+1, \ldots, p+r-1$ such that
    $$\sigma < \alpha_1 < \alpha_2 < \ldots < \alpha_{r-1} < \alpha.$$
 By definition of $f$, $f(\sigma)\preceq f(\alpha_h)\preceq f(\alpha)$ for $h=1,\ldots r-1$. Recalling that $f(\alpha)=f(\sigma)$ we see that $f(\alpha_1) = \ldots = f(\alpha_{r-1}) = f(\sigma)$. Thus $\alpha_h\in L_*(\sigma)$ for $h=1,\ldots r-1$. In particular, $\alpha_1$ is a cofacet of $\sigma$ that belongs to $L_*(\sigma)$. Every cofacet of $\sigma$ in $L_*(\sigma)$ has only $\sigma$ as unclassified facet in $L(\sigma)$ by Lemma~\ref{lem:correctness1}. Let $\delta$ be the cofacet of $\sigma$ with minimal index. Statement {\em (i)} implies that $\delta$ is still unclassified.

\item[{\em (iii)}] Let $\dim\sigma=p$ and $\dim\alpha=p+r$. If $r>1$ then there are at least two sequences $\sigma<\alpha_1<\ldots <\alpha_{r-1}<\alpha$ and $\sigma<\alpha_1'<\ldots <\alpha_{r-1}'<\alpha$ of cells belonging to  $L(\sigma)$ with $\alpha_{r-1}\ne \alpha_{r-1}'$. These cells $\alpha_{r-1}$ and  $\alpha_{r-1}'$ need to be already classified at line 11 because of the assumption $\texttt{{num}\_{unclass}\_{facets}}_{\sigma} (\alpha)=0$.  By {\em (i)}, they had not been classified when $i<I(\sigma)$. Since we are at line 11, it has necessarily occurred when $i=I(\sigma)$ at line 9. But the coface $\delta$ of $\sigma$ with minimal index is unique so only one between $\alpha_{r-1}$ and  $\alpha_{r-1}'$ has been classified at line 9, giving a contradiction. Thus, $r=1$.
\end{enumerate}
\end{proof}

\begin{lem}\label{lem:line19}
 Let $\alpha \in L_* (\sigma)$ be such that when it is popped from \texttt{PQone}
at line $15$ of Algorithm~\ref{alg-match}, \texttt{{num}\_{unclass}\_{facets}}$_{\sigma} (\alpha$) = 1. Then the unique cell
$\lambda \in \texttt{{unclass}\_{facets}}_{\sigma} (\alpha)$ belongs to \texttt{PQzero}. Therefore all cells popped out from \texttt{PQone} at line $15$ of Algorithm~\ref{alg-match} for which \texttt{{num}\_{unclass}\_{facets}}$_{\sigma} (\alpha$) = 1 get paired at line 19.
\end{lem}

\begin{proof}
We reason by induction on $r \geq 2$ where $\dim \alpha = p+r$. Note that for $r=1$, $\alpha$ is a cofacet of $\sigma$ with 0 unclassified faces in $L(\sigma)$ after step 9 is executed. Therefore $\alpha$ cannot enter \texttt{PQone}. For $r=2$, $\lambda$ is a primary facet of $\sigma$ with 0 unclassified faces in $L(\sigma)$ after step 9 is executed. Therefore $\lambda \in \texttt{PQzero}$. Assume by induction that for each natural number $j$ from 2 up to value $r-1$, when $\alpha$ with $\dim \alpha = p+ j$ is popped from \texttt{PQone} with \texttt{{num}\_{unclass}\_{facets}}$_{\sigma} (\alpha$) = 1, its unique unclassified face $\lambda$ belongs to \texttt{PQzero} and therefore $\alpha$ and $\lambda$ get paired at line 19. Let now $j=r$, and  $\dim\alpha=p+j$. Let us assume that $\lambda$ is not in \texttt{PQzero}. Then there are two cases. If $\lambda$ has entered \texttt{PQone}, then it has been processed before $\alpha$. Since $\lambda$ is not in \texttt{PQzero}, by the induction hypothesis, it must have been paired with some cell in \texttt{PQzero}. This is a contradiction to the statement \texttt{{num}\_{unclass}\_{facets}}$_{\sigma} (\alpha$) = 1. If $\lambda$ did not enter \texttt{PQone}, then the number of unclassified faces of $\lambda$ in $L(\sigma)$ is greater than or equal to 1. Thus, since $\lambda$ is of dimension $p+r-1$, there must exist a face $\tau^{(p+r-2)}$ of $\lambda$ of dimension $p+r-2$ in $L_* (\sigma)$ that is not paired and not added to $\sC$. This process can be carried out until we get a $(p+1)$-cell $\tau^{(p+1)}$ in $L_*(\sigma)$ that is not classified by the algorithm. In general, we get sequences in $L(\sigma)$ such that
$$\sigma^{(p)}<\tau^{(p+1)}<\tau^{(p+2)}<\ldots < \tau^{(p+r-2)}<\lambda ^{(p+r-1)} <\alpha^{(p+r+2)}$$ with $\tau^{(p+1)}$ not classified by the algorithm. By Lemma~\ref{lem:correctness1} the number of unclassified faces of $\tau^{(p+1)}$ is 0, implying that $\tau^{(p+1)}$ has entered \texttt{PQzero}. Let us fix the sequence for which $\tau^{(p+2)}$ is of minimal index and has only one unclassified face, hence it has entered \texttt{PQone} before $\alpha$. It exists because $\delta$ has been classified and $\cK$ is a simplicial complex.  We deduce that $\tau^{(p+1)}$  and  $\tau^{(p+2)}$ have been paired, contradicting the assumption that $\tau^{(p+1)}$ has not  been classified by the algorithm. Hence, $\lambda$ should belong to \texttt{PQzero}, which completes the proof.
\end{proof}

\begin{lem}\label{lem:correctness3}
Let $\sigma\in \cK$. Each cell in $L(\sigma)$ is processed exactly once by the algorithm
and it is paired with some other cell or classified as critical. Hence Algorithm~\ref{alg-match} classifies all cells of $\cK$ and always terminates.
\end{lem}

\begin{proof} We break the conclusion into three statements:
\begin{aenum}
\item Each cell in the lower star eventually enters PQone or PQzero.
\item Each cell that has entered PQone or PQzero is eventually classified.
\item A cell that has already been classified cannot enter PQone or PQzero again.
\end{aenum}

We simultaneously prove the three statements by induction on $i=I(\sigma)$. For $i=1,2$ the claim is proved by Lemma~\ref{lem:i=1-2}. Let us now assume by induction that the claim is true from 2 up to $i-1$. Let $I(\sigma)=i$. If \texttt{classified($\sigma$)}={\bf true}, then $\sigma$ has already been classified as part of $L(\beta)$ for some cell $\beta$ that is a face of $\sigma$. Thus, $I(\beta) < I(\sigma) = i$ and $L(\sigma) \subset L(\beta)$. By induction hypothesis, every cell of $L(\sigma)$ is processed once by the algorithm
and it is paired with some other cell or classified as critical. If \texttt{classified($\sigma$)}={\bf false}, $\sigma$ is  either declared critical at line 7 or, by Lemma~\ref{lem:line10}{\em (ii)}, paired with some other cell $\delta$ in $L_* (\sigma)$ at line 10. The cells $\sigma$ and $\delta$ are no further processed.

Let $\gamma$ be a cell left in $L_*(\sigma)$, if any. Suppose that
\begin{equation}\label{eq:dimleq1}
\texttt{{num}\_{unclass}\_{facets}}_{\sigma}(\gamma)\leq 1.
\end{equation}
Then $\gamma$ is either added to \texttt{PQzero} or to \texttt{PQone} and it is ultimately either paired or classified as critical. More precisely, if $\gamma$ is added to \texttt{PQone}, then it is either moved to \texttt{PQzero} at line 17, or paired at line 19 by Lemma~\ref{lem:line19}. If $\gamma$ is added to \texttt{PQzero}, it is either paired at line 20 or declared critical at line 26. This also shows that, when $i=I(\sigma)$, every cell in $L_* (\sigma)$ enters at most once in \texttt{PQzero} and \texttt{PQone}.

It remains to show that (a) also holds for cells $\gamma$ with
\begin{equation}\label{eq:dimgeq2}
\texttt{{num}\_{unclass}\_{facets}}_{\sigma}(\gamma)\geq 2.
\end{equation}
We prove (a) by induction on the dimension $n$ of cells in $L_*(\sigma)$. The initial step is $\dim\gamma=\dim\sigma+1$. But then $\gamma$ is a cofacet of $\sigma$ and, by Lemma~\ref{lem:correctness1}, (\ref{eq:dimleq1}) holds. Assume by induction that all cells of $L_*(\sigma)$ with dimension smaller than $n$ have entered either \texttt{PQzero} or \texttt{PQone}.

Let $\gamma$ be a cell of dimension $n$ in $L_* (\sigma)$. If (\ref{eq:dimleq1}) holds, we are done. Suppose that (\ref{eq:dimgeq2}) holds. We show that $\gamma$ eventually enters \texttt{PQone}. By induction, all faces of $\gamma$ eventually enter \texttt{PQzero} or \texttt{PQone}. We have earlier shown that those which enter \texttt{PQone} are classified or moved to \texttt{PQzero}. So all faces of $\gamma$ which are not classified enter \texttt{PQzero}. All such faces which have a coface in \texttt{PQone} or get a coface in \texttt{PQone} at line 21 are classified at line 19-20. We remain with the faces of $\gamma$ which are in \texttt{PQzero} but have no coface in \texttt{PQone}. Let $r$ be the number of such faces. Necessarily $r>1$, otherwise $\gamma$ is in \texttt{PQone}. At lines 25-26 one of those faces is classified as critical, so we remain with $r-1$ such faces.  After passing other $r-2$ times through lines 25-26, $\gamma$ remains with only one unclassified face and it is added to \texttt{PQone} at line 27.

So we have proved that every cell in $L(\sigma)$ is processed exactly once by the algorithm while $i=I(\sigma)$ and it is paired with some other cell of $L(\sigma)$ or classified as critical.

Finally, since the number of cells in the complex $\cK$ is finite and the union of $L (\sigma)$'s covers the complex, the proof is complete.
\end{proof}

%%%%%%%%%%%%%%%%%%%%%%%%%%%%%%%%%%%%%%%%%%%%%%%%%%%%%%%%%%%%%%%%%%%%%%%%%%%%%%%%%%%%%%%%%%%%%%%%%%%%%%%%%%%%%%%%%%%%%%%%%%%%%%%%%%%%%%%%%%%%%%%
\begin{prop}\label{lem:partition}
 $\sA, \sB, \sC$ is a partition of the complex ${\cK}$ and $\ma$ is a bijective function from $\sA$ to $\sB$.
Moreover, if $\sigma\in \cK^\alpha \cap \sA$ then $\ma(\sigma)\in \cK^\alpha$.
\end{prop}

\begin{proof}
By Lemma \ref{lem:correctness3},  $\sA \cup \sB \cup \sC= {\cK}$. We show that $\sA \cap \sB = \emptyset$.
This statement is trivial for vertices since they cannot belong to $\sB$. Assume on the contrary that there exists a cell $\alpha^{(p)}$ with $p \geq 1$ such that
$\alpha \in \sA \cap \sB$. Thus, there exist cells $\delta^{(p-1)} < \alpha < \gamma^{(p+1)}$ such that
$\ma (\alpha) = \gamma $ and $\ma (\delta) = \alpha$. This means that $\alpha$ is paired twice by processing two different lower stars $L(\sigma_1)$ and $L(\sigma_2)$.
By Lemma~\ref{lem:inj}(4), there exists a cell $\beta$ such that the cells $\delta, \alpha$ and $\gamma$ are all processed within $L(\beta)$. Thus $\alpha$ is
processed twice within $L(\beta)$, which contradicts Lemma~\ref{lem:correctness3}.  If we assume that $\sA \cap \sC \neq \emptyset$
and contains a cell $\alpha$, then $\alpha$ has been declared critical either at line 7 or at line 26. In the first case,  $\alpha$ was not previously assigned to $\sA$
because of line 5; on the other hand it cannot be assigned to $\sA$ later because of Lemma~\ref{lem:correctness3}. In the second case,  when $\alpha$ is added to $\sC$,
it comes from $\texttt{PQzero}$ and $\texttt{PQone}$ is empty. The only cells that may enter $\texttt{PQzero}$ or $\texttt{PQone}$ later are cofaces of $\alpha$ (see line 27).
Therefore $\alpha$ cannot be added again to $\texttt{PQzero}$, and as a consequence it cannot be added to $\sA$. The proof that   $ \sB \cap \sC = \emptyset$ can be handled in
much the same way. It follows that $\sA, \sB, \sC$ is a partition of $\cK$.

By construction, the map $\ma$ is onto. We will show that $\ma$ is injective. If two cells $\sigma_1$ and $\sigma_2$ are paired with the same cell $\alpha$, it follows that $\alpha$
must belong to the intersection of two lower stars. Therefore, again by Lemma~\ref{lem:inj}(4), there must exist a $\beta$ such that $\alpha$ is processed twice by the algorithm within
$L(\beta)$ which is again a
contradiction to Lemma~\ref{lem:correctness3}. Thus $\ma$ is bijective.

By construction, $\sigma$ is a face of $\ma (\sigma)$ and they both belong to some $L(\beta)$. Thus, by Lemma~\ref{lem:inj}(1),
$f(\sigma) = f(\ma(\sigma))$, and therefore if $\sigma \in \cK^\alpha \cap \sA$ then $\ma(\sigma)\in \cK^\alpha$.

\end{proof}

\begin{thm}\label{th:acyclicity}
Algorithm~\ref{alg-match} produces a partial matching $(\sA, \sB, \sC, \ma)$ that is acyclic.
\end{thm}

\begin{proof}
A partial matching is acyclic if and only if there are no nontrivial closed $\ma$--paths as defined in (\ref{eq:m-path}). We prove this by contradiction. Assume that
\begin{equation}\label{eq:directed-loop}
\sigma_0 \xrightarrow{\ma}  \tau_0 \xrightarrow{>} \sigma_1 \xrightarrow{\ma} \tau_1 \xrightarrow{>} \ldots \xrightarrow{>} \sigma_n \xrightarrow{\ma}
\tau_n \xrightarrow{>} \sigma_0
\end{equation}
is a directed loop in the modified Hasse diagram. In particular, all $\sigma_i$ in the loop have the same dimension, say, $p$ and all $\tau_i$ have the same dimension $p+1$.
The index $i$ of $\sigma_i$ is not the value of the indexing function $I$ but it simply displays its position in the loop.
From Lemma~\ref{lem:inj}(1), it follows that
\begin{equation}\label{eq:directed-loop-f}
f(\sigma_0) =  f(\tau_0)  \succeq  f(\sigma_1) = f(\tau_1) \succeq  \ldots \succeq  f(\sigma_n) = f( \tau_n) \succeq f(\sigma_0).
\end{equation}
If any of the inequalities $f(\tau_{i-1})  \succeq  f(\sigma_i)$ is strict, then there exists a coordinate $j$ such that
$f_j(\tau_{i-1}) > f_j(\sigma_i)$ and so $f_j(\sigma_0) > f_j(\sigma_0)$, a contradiction. Hence $f$ is constant on all the elements of the loop.
Let us set $\bar \sigma$ equal to the cell such that $$I(\bar\sigma) =\min\left\{I(\alpha)\in \N: \alpha\subseteq \bigcap_{i=0}^n\sigma_i\cap \bigcap_{i=0}^n\tau_i\right\}.$$
The simplex $\bar\sigma$ exists  by Lemma~\ref{lem:inj}(3). This implies that $\sigma_i$ and $\tau_i $ belong to $L(\bar \sigma)$ for $i=0,\ldots, n$. Now we have two cases: either
$\bar\sigma=\sigma_j$ for some $j$, $0\le j\le n$, or $\bar \sigma \subset \sigma_j$ for every $j$. In the first case, without a loss of generality, we may assume that $\bar\sigma=\sigma_0$.  Since $\sigma_n$ has the same dimension as $\sigma_0$, it is  in $L(\sigma_0)$ if and only if $\sigma_n=\sigma_0$, implying that the loop is trivial, a contradiction.
In the second case, note that Algorithm~\ref{alg-match} produces a pairing $\ma(\sigma_i)=\tau_i$ only when \texttt{{num}\_{unclass}\_{facets}}$_{\bar \sigma} (\tau_i)$ = 1, and in that case the unclassified face of $\tau_i$ is exactly $\sigma_i$. Therefore, we have that $\sigma_0$ is paired to $\tau_0$ after that $\sigma_1$,  also a face of $\tau_0$, has been paired to $\tau_1$.

Iterating this argument for $i=1,\ldots, n$, we deduce that $\sigma_0$ is paired to $\tau_0$ after that $\sigma_n$ has been paired to $\tau_n$. But since $\sigma _0$ is also a face of $\tau_n$, and $\sigma_0$ is still unclassified when $\sigma_n$ is paired to $\tau_n$, it follows that $\sigma_n=\sigma_0$, implying that the loop is trivial, again a contradiction.

\end{proof}

\subsection{Complexity analysis}
\label{sec:complexity}

{\bf Definitions and parameters.} We use the following definitions and parameters in estimating the computational cost of Algorithm~\ref{alg-match}.
\begin{enumerate}
 \item Given a simplex $\sigma \in \cK$, the coboundary cells of $\sigma$ are given by
$$\cb (\sigma) := \{ \tau \in \cK \, |\, \sigma \, \mbox{is a face of } \,  \tau\}.$$
It is immediate from the definitions that $L_* (\sigma) \subset \cb (\sigma)$.

\item We define the coboundary mass $\gamma$ of $\cK$ as
$$\gamma = \max_{\sigma \in \cK} \card \cb(\sigma),$$
where $\card$ denotes cardinality. While $\gamma$ is trivially bounded by $N$, the number of cells in $\cK$,
this upper bound is a gross estimate of $\gamma$ for many complexes of manifolds and approximating surface boundaries of objects.

\item For the simplicial complex $\cK$, we assume that the boundary and coboundary cells of each simplex
are computed offline and stored in such a way that access to every cell is done in constant time.

\item Given an admissible function $f : {\cK} \to \R^k$, the values by $f$ of simplices $\sigma \in \cK$ are stored in the structure
that stores the complex $\cK$ in such a away that they are accessed in constant time.

\item We assume that adding cells to the lists $\sA$, $\sB$, and $\sC$ is done in constant time.
\end{enumerate}

Algorithm~\ref{alg-match} processes every cell $\sigma$ of the simplicial complex $\cK$ and checks whether it is classified or not.
In the latter case, the algorithm requires a function that returns the cells in the reduced lower star $L_* (\sigma)$ which
is read directly from the structure storing the complex. In the best case, $L_* (\sigma)$ is empty and the cell is declared critical.
Since $L_* (\sigma) \subset \cb (\sigma)$, it follows that $\card L_* (\sigma) \leq \gamma$.
As stated earlier in proof of Lemma~\ref{lem:correctness3} every cell in $L_* (\sigma)$ enters at most once in \texttt{PQzero} and \texttt{PQone}.
It follows that the while loops in the algorithm are executed all together in at most $2 \gamma$ steps. We may consider the operations such as finding
the number of unclassified faces of a cell to have constant time except for the priority queue operations which are logarithmic in the size of the priority queue
when implemented using heaps. Since the sizes of \texttt{PQzero} and \texttt{PQone} are clearly bounded by $\gamma$, it follows that $L_* (\sigma)$ is
processed in at most $O(\gamma \log \gamma)$ steps.
Therefore processing the whole complex incurs a worst case cost of $O(N \cdot \gamma \log \gamma)$.

%%%%%%%%%%%%%%%%%%%%%%%%%%%%%%%%%%%%%%%%%%%%%%%%%%%%%%%%%%%%%
%%%%%%%%%%%%%%%%%%%%%%%%%%%%%%%%%%%%%%%%%%%%%%%%%%%%%%%%%%%%%
%
%

\section{Persistent Homology Reduction}
\label{per-hom}

\subsection{Lefschetz complexes}

For the purpose of the Matching Algorithm~\ref{alg-match}, we only needed simplical complexes but, for its applications to computing persistent homology, one shall need a more general class of cellular complexes. A convenient combinatorial framework for a cellular complex is that of a Lefschetz complex introduced by Lefschetz in \cite{Lef42} under the name {\em complex}. The properties of Lefchetz complexes are developed further with the purpose of efficient homology computation in \cite{MrBa09} under the name $\sS$-complex and we refer to this term in our first paper \cite{AlKaLa17}.

A {\em  Lefschetz complex}, equivalently, {\em $\sS$-complex} is a graded set $\sS=\{\sS_q\}_{q\in\Z}$ of elements which we shall call {\em cells} with a {\em facet relation} $\tau < \sigma$ and {\em incidence numbers} $\kappa(\sigma, \tau)$ which are zero unless $\tau$ is a facet of $\sigma$. The incidence numbers are defined so that they give rise to a chain boundary map $\bdy^\kappa_q:C_q(\sS) \to C_{q-1}(\sS)$, thus defining a free chain complex $C_*(\sS, \bdy^\kappa)$. It is assumed in \cite{MrBa09} that the chain coefficients are in a principal ideal domain $R$ but in our paper we assume that $R=F$ is a field. A simplicial complex $\cK$ is a particular case of a Lefschetz complex. For the incidence relations, one needs to impose an orientation on simplices of $\cK$ unless we compute the homology with $\Z_2$ coefficients.

By the homology of a Lefschetz complex $(\sS, \kappa)$ we mean the homology of the chain complex $(C_*(\sS),\partial^\kappa_* )$, and we denote it by $H_*(\sS, \kappa)$ or simply by $H_*(\sS)$.

\subsection{Multifiltration on Lefschetz complexes}\label{sec:S-m-filt}

The concept of {\em sublevel set filtration} of $\cK$ induced by $f: \cK \to \R^k$ introduced in Section~\ref{sec:md-f} naturally extends to Lefschetz complexes as follows.

Let $(\sS, \kappa)$ be a Lefschetz complex. A {\em multi-filtration} of $\sS$ is a family $\cF=\{\sS^\alpha\}_{\alpha\in \R^k}$ of
subsets of $\sS$ with the following properties:
\begin{aenum}
\item $\cF$ is nested with respect to inclusions, that is $\sS^\alpha\subseteq \sS^{\beta}$, for every $\alpha\preceq \beta$,
where $\alpha\preceq \beta$ if and only if $\alpha_i \leq \beta_i$ for all $i=1,2,\ldots, k$;
\item $\cF$ is non-increasing on faces, that is, if $\sigma \in \sS^\alpha$ and $\tau$ is a face of $\sigma$ then $\tau \in \sS^\alpha$.
\end{aenum}

Persistence is based on analyzing the homological changes occurring along the filtration as $\alpha$ varies. This analysis is carried out by considering, for $\alpha\preceq\beta$, the homomorphism
\[
H_*(j^{(\alpha,\beta)}): H_*(\sS^{\alpha}) \to H_*(\sS^{\beta}).
\]
induced by the inclusion map $j^{(\alpha,\beta)}:\sS^{\alpha}\hookrightarrow \sS^{\beta}$.

The image of the map $H_q(j^{(\alpha,\beta)})$ is  known as the {\em $q$'th multidimensional persistent homology group} of the filtration at $(\alpha,\beta)$ and we denote it by $H_q^{\alpha,\beta}(\sS)$. It contains the homology classes of order $q$ born not later than $\alpha$ and still alive at $\beta$.

\subsection{Reductions}\label{sec:reductions}

The definition of a partial matching given in Section~\ref{sec:match} extends in a straightforward way to any Lefschetz complex $(\sS, \kappa)$~\cite{AlKaLa17}. The only substantial difference is that the condition ``for each $\tau\in \sB$, $\ma(\tau)$ is a cofacet of $\tau$'' needs to be replaced by the condition ``for each $\tau\in \sB$, $\kappa(\ma(\tau),\tau)$ is invertible'', unless field coefficients are assumed.

Let  $(\sA,\sB,\sC,\ma)$ be a partial matching (not necessarily acyclic) on a Lefschetz complex $(\sS, \kappa)$. Given $\sigma\in \sA$, a new Lefschetz complex $(\overline{\sS}, \overline{\kappa})$ is constructed in \cite{AlKaLa17} by setting $\overline{\sS}=\sS\setminus \{\ma(\sigma),\sigma\}$, and $\overline{\kappa}:\overline{\sS}\times \overline{\sS}\to R$,
\begin{equation}\label{eq:kappa-bar}
\overline{\kappa}(\eta,\xi)=\kappa(\eta,\xi)-\frac{\kappa(\eta,\sigma)\kappa(\ma(\sigma),\xi)}{\kappa(\ma(\sigma),\sigma)}.
\end{equation}
We say that $(\overline{\sS}, \overline{\kappa})$ is obtained from $(\sS,\kappa)$ by a {\em reduction} of the pair $(\ma(\sigma),\sigma)$.

It is well known \cite{KaMrSl98} that $C_*(\overline{\sS})$ is a well-defined chain complex and that there exists a pair of linear maps $\pi: C_*(\sS)\to C_*(\overline{\sS})$ and $\iota: C_*(\overline{\sS})\to C_*(\sS)$, explicitly defined by \cite[Formulas (3,4)]{AlKaLa17}, which are chain equivalences. As a consequence,
\begin{equation}\label{eq:H-S=Sbar}
H_*(\overline{\sS})\cong H_*(\sS).
\end{equation}

Let  $(\sA,\sB,\sC,\ma)$ be an acyclic partial matching  on a Lefschetz complex $(\sS, \kappa)$. Let $(\overline{\sS}, \overline{\kappa})$ be obtained from $(\sS,\kappa)$ by reduction of the pair $(\ma(\sigma),\sigma)$, $\sigma\in \sA$. It is proved in \cite{AlKaLa17} that $\overline{\kappa}(\ma(\tau),\tau)=\kappa(\ma(\tau),\tau)$ for any $\tau\in \sA\setminus \{\sigma\}$.

The Algorithm~\ref{alg-match} takes a simplicial complex $\cK$ as the input. Let us set
\[
\sS(0)=\cK
\]
as the initial complex. If we apply a reduction to $\sS(0)$, we may get a Lefchetz complex which is no longer a simplicial complex. However, as a set of generators of a chain complex, it is the subset of the original one: it consists of simplices of $\cK$ and only the incidence numbers get changed. This fact is used to prove in \cite{AlKaLa17} that one-step reductions can be iterated on new subcomplexes. We do not iterate the matching algorithm: it is applied only once to the initial simplicial complex $\cK$. We iterate reductions of pairs initially matched by the algorithm. Here are the steps towards the main result on persistent homology of $\sS$ proved in \cite{AlKaLa17}. All the following statements equally apply to the matching produced by the new Algorithm~\ref{alg-match}.

\begin{prop}\label{prop:reduced-matching}
Let  $(\sA,\sB,\sC,\ma)$ be an acyclic partial matching on $(\sS,\kappa)$. Given a fixed $\sigma\in \sA$,
define $\overline{\sA}=\sA\setminus \{\sigma\}$,
$\overline{\sB}=\sB\setminus \{\ma(\sigma)\}$, $\overline{\ma}=\ma_{|\overline{\sA}}$, and $\overline{\sC}=\sC$.
Then $(\overline{\sC},\overline{\ma}:\overline{\sA}\to \overline{\sB})$ is an acyclic partial matching on $(\overline{\sS},\overline{\kappa})$.
\end{prop}

Let ${\cF}=\{\sS^\alpha\}_{\alpha\in\R^k}$ be a multifiltration on $\sS$. Then $\overline{\cF}=\{\overline{\sS}^\alpha\}_{\alpha\in\R^k}$ is the {\em induced multifiltration} on
$\overline{\sS}$ defined by setting, for each $\tau \in \overline{\sS}$,
\[
\tau \in \overline{\sS}^\alpha \iff \tau\in \sS^\alpha.
\]

In the sequel, we assume that $(\sA,\sB,\sC,\ma)$ is an acyclic matching on a filtered Lefschetz complex $\sS$ with the property:

\begin{equation}\label{matching-filtration}
\mbox{If } \sigma\in\sS^\alpha \mbox{ then } \ma(\sigma)\in \sS^\alpha.
\end{equation}

Lemma~\ref{lem:partition} asserts that the matching produced by Algorithm~\ref{alg-match} on a filtered simplicial complex $\sS(0)=\cK$ has this property.

\begin{lem}\label{lem:chain-equivalence} Let $\sigma\in \sA$ and let $(\overline{\sS}, \overline{\kappa})$ be obtained from $(\sS,\kappa)$ by reduction of the pair $(\ma(\sigma),\sigma)$.
Let $\pi$ and $\iota$ be chain equivalences defined by the formulas (3) and (4) in \cite{AlKaLa17} respectively.
Then the maps $\pi_{|C_*(\sS^\alpha)}:C_*(\sS^\alpha)\to C_*(\overline{\sS}^\alpha)$ and
$\iota_{|C_*(\overline{\sS}^\alpha)}:C_*(\overline{\sS}^\alpha)\to C_*(\sS^\alpha)$ defined by restriction are chain homotopy equivalences.
Moreover, the diagram
\[
\begin{array}{ccc}
H_*(\sS^\alpha) & \mapright{H_*(j^{(\alpha,\beta)})} & H_*(\sS^\beta) \\
\mapdown{\cong} & & \mapdown{\cong} \\
H_*(\overline{\sS}^\alpha) & \mapright{H_*(j^{(\alpha,\beta)})} & H_*(\overline{\sS}^\beta)
\end{array}
\]
commutes.
\end{lem}

Lemma~\ref{lem:chain-equivalence} immediately yields the following result.

\begin{thm}\label{th:reduced-filtration-iso}
For every $\alpha\preceq \beta\in \R^k$, $H_*^{\alpha,\beta}(\overline{\sS}) \cong H_*^{\alpha,\beta}(\sS)$.
\end{thm}

We now let $\sS=\cK$ and consider the matching on $\cK$ produced by Algorithm~\ref{alg-match}. We order $\sA$ in a sequence
\[
\sA=\{\sA(1),\sA(2),\ldots,\sA(n)\}
\]
and set $\sB(i)=\ma(\sA(i))$, $i=1,\ldots,n$. Put $\sS(0)=\sS=\cK$ and
\[
\sS(i)=\overline{\sS(i-1)}=\sS(i-1)\setminus \{\sB(i),\sA(i)\},\;\;i=1,2,\ldots,n.
\]
Since a partial matching defines a partition of $\sS$, we have $\sS(n)=\sC$.

By definition of induced filtration, the condition (\ref{matching-filtration}) carries through to the reduced complex. Consequently, Proposition~\ref{prop:reduced-matching}, Lemma~\ref{lem:chain-equivalence} and Theorem~\ref{th:reduced-filtration-iso} extend by induction to any step of reduction. Hence, for any $\alpha\in \R^k$,  we get a sequence of filtered Lefschetz complexes
\[
(\sS^{\alpha}(0),\kappa^{\alpha}(0)),\; (\sS^{\alpha}(1),\kappa^{\alpha}(1)),\;\ldots,\;(\sS^{\alpha}(n),\kappa^{\alpha}(n)),
\]
where $\kappa^{\alpha}(i)=\overline{\kappa^{\alpha}(i-1)}$, together with a sequence of chain equivalences
\[
\pi^{\alpha}(i): C_*(\sS^{\alpha}(i-1))\to C_*(\sS^{\alpha}(i)),\;\;\iota^{\alpha}(i): C_*(\sS^{\alpha}(i))\to C_*(\sS^{\alpha}(i-1)).
\]
Moreover, for any $\alpha\preceq \beta$, we get the sequence of inclusions
\[
j^{(\alpha,\beta)}(i):\sS^{\alpha}(i)\hookrightarrow \sS^{\beta}(i),
\]
such that the commutative diagram of Lemma~\ref{lem:chain-equivalence} can be applied to the $i$'th iterate. By induction, we get the final result.

\medskip

\begin{cor}\label{cor:homology-S-iso-C}
For every $\alpha\preceq \beta\in \R^k$, $H_*^{\alpha,\beta}(\sC) \cong  H_*^{\alpha,\beta}(\cK)$.
Moreover, the diagram
\[
\begin{array}{ccc}
H_*(\cK^\alpha) & \mapright{H_*(j^{(\alpha,\beta)})} & H_*(\cK^\beta) \\
\mapdown{\cong} & & \mapdown{\cong} \\
H_*(\sC^\alpha) & \mapright{H_*(j^{(\alpha,\beta)})} & H_*(\sC^\beta)
\end{array}
\]
commutes.
\end{cor}

Corollary~\ref{cor:homology-S-iso-C} asserts that the multidimensional persistent homology of the reduced complex is the same as of the initial complex.

An equivalent reduction procedure for linearly filtered complexes is presented in \cite{MiNa}. Its implementation and complexity analysis can be applied in our setting. It is shown there that the worst case computational complexity of the reduction procedure is bounded by the product $O(N \gamma m^2)$, where $N$ is the number of cells in $\cK$, $\gamma$ is the bound on the number of cofaces of any given cell, and $m$ is the cardinality of the set $\sC$. In many practical situations, meshing techniques produce regular tirangulations, where $\gamma$ can be assumed constant. If, moreover, the partial matching algorithm produces $\sC$ significantly smaller than $\cK$, the complexity as a function of $N$ is close to linear.

Recently, new reduction techniques were designed among others by \cite{DloWag} and \cite[Section C.2]{Rat15}. We believe that our Matching Algorithm~\ref{alg-match} can also be applied together with those techniques to speed up the multidimensional persistent homology computation.

%
%
%%%%%%%%%%%%%%%%%%%%%%%%%%%%%%%%%%%%%%%%%%%%%%%%%%%%%%%%%%%%%
%%%%%%%%%%%%%%%%%%%%%%%%%%%%%%%%%%%%%%%%%%%%%%%%%%%%%%%%%%%%%

\section{Experimental Results}\label{sec:experiments}

We have successfully applied the algorithms from Section \ref{sec:alg}
to different sets of triangle meshes. In this
section we present two numeric applications - on synthetic and on real data - for which
an acyclic matching preserving multidimensional persistent homology is computed.

In each case the input data is a 2-dimensional simplicial complex $\cK$  and a function $f$ defined on the vertices of $\cK$ with values in $\R^2$.

The first step is to slightly perturb $f$ in order to achieve injectivity on each component as described in Section \ref{sec:prel}. The second step is to construct an index function defined on all the simplices of the complex and satisfying the properties of Lemma \ref{lem:index}. Then we build
the acyclic matching $\ma$ and the partition $(\sA,\sB,\sC)$ in the simplices of the complex using Algorithm \ref{alg-match}. In particular, the number of simplices in $\sC$ out of the total number of simplices of $\cK$ is relevant, because it determines the amount of reduction obtained by our algorithm to speed up the computation of multidimensional persistent homology.

\subsection{Examples on Synthetic Data}

We consider three well known 2-dimensional manifolds - the sphere, the torus, and the Klein bottle - triangulated in different ways.

In the case of spheres, we consider triangulations of five different sizes and we take
$f(x,y,z)=(x,y)$. The triangulated sphere with the least number of simplices is shown in Figure \ref{fig:sphere}, together with the acyclic matching produced by our algorithm (critical cells are in red) and the corresponding $\ma$--paths.

\begin{figure}[t]
\begin{center}
\includegraphics[width=0.80\textwidth]{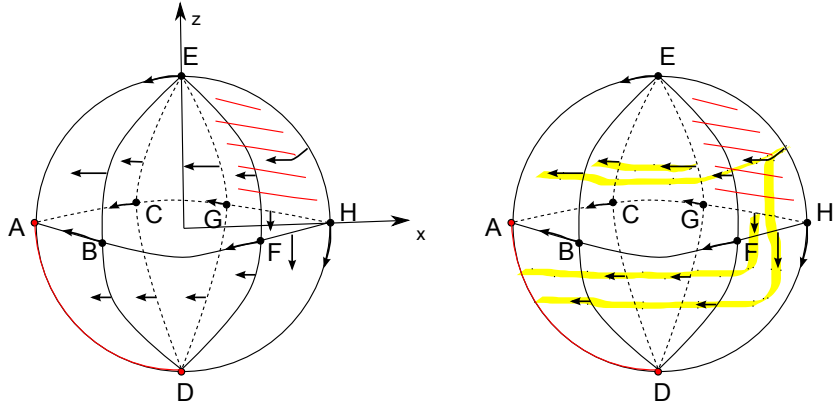}
  \caption{Left: A triangulated sphere $S^2$ in $\R^3$ with the acyclic matching and  critical cells (in red)  associated with a perturbation of the function $f$ given on its vertices by $f(x,y,z)=(x,y)$. Right: the corresponding $\ma$-paths (in yellow).
 }
\end{center}
\label{fig:sphere}
\end{figure}

The comparison with other triangulations of the sphere is shown in Table \ref{tab:sphere}:  the first row  shows  the number of  simplexes in each considered mesh $\cK$; the middle row shows the number of critical cells  obtained by using our matching algorithm to reduce $\cK$; the bottom row shows the ratio between the second and the first lines, expressing them in percentage points.

\begin{table}[h]
\caption{Reduction performance on five different triangulations of the sphere. }
\begin{center}
\begin{tabular}{| c | c | c | c | c |c|}
\hline
 & \tt{sphere\_1}&
 \tt{sphere\_2} & \tt{sphere\_3}
 &  \tt{sphere\_4} & \tt{sphere\_5} \\
 \hline
  $\begin{array}{rrr}\# \cK\\ \# \sC\\ \% \end{array}$ &
 $\begin{array}{rrr}  38 \\4 \\10.5263  \end{array}$ &
 $\begin{array}{rrr}  242\\20 \\8.2645  \end{array}$ &
 $\begin{array}{rrr}   962   \\98 \\10.1871  \end{array}$ &
 $\begin{array}{rrr}  1538    \\178 \\11.5735 \end{array}$&
 $\begin{array}{rrr}  2882    \\278 \\9.6461  \end{array}$\\
 \hline
 \end{tabular}
 \end{center}
\label{tab:sphere}
\end{table}

In the case of the torus, we again consider triangulations of different sizes and we take
$f(x,y,z)=(x,y)$. The numerical results are shown in Table \ref{tab:torus} (see also Figure \ref{fig:torus}).

\begin{table}[h]
\caption{Reduction performance on different triangulations of the torus.}
\begin{center}
\begin{tabular}{| c | c | c | c|}
\hline
 & \tt{torus\_96}&
 \tt{torus\_4608} & \tt{torus\_7200} \\
 \hline
  $\begin{array}{rrr}\# \cK\\ \# \sC\\ \% \end{array}$ &
 $\begin{array}{rrr}  96 \\8 \\8.3333  \end{array}$ &
$\begin{array}{rrr}  4608 \\128 \\ 2.7778  \end{array}$ &
$\begin{array}{rrr}  7200 \\156 \\  2.1667  \end{array}$
\\
 \hline
 \end{tabular}
 \end{center}
\label{tab:torus}
\end{table}

\begin{figure}[ht]
\begin{center}
\includegraphics[width=0.80\textwidth]{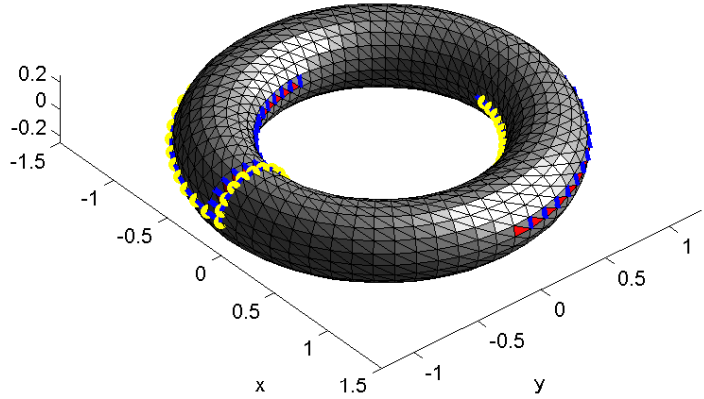}
  \caption{A triangulation of the torus with 7200 simplices, with respect to a component-wise injective perturbation of the function defined of its vertices by $f(x,y,z)=(x,y)$, has 156 critical simplices: 39 critical vertices  out of 1200  (in yellow), 78 critical edges  out of 3600 (in blue),  and 39 critical faces  out of 2400 (in red).
 }
\end{center}
\label{fig:torus}
\end{figure}

Finally, similar tests in the case of an approximation of a Klein bottle immersed in $\R^3$ give the results displayed in Table \ref{tab:klein} and Figure \ref{fig:klein}.

\begin{table}[h]
\caption{Reduction performance on different triangulations approximating an immersion of the Klein bottle.}
\begin{center}
\begin{tabular}{| c | c | c | c| c |}
\hline
 & \tt{klein\_89}&
 \tt{klein\_187} & \tt{klein\_491} &  \tt{klein\_1881}\\
 \hline
  $\begin{array}{rrr}\# \cK\\ \# \sC\\ \% \end{array}$ &
 $\begin{array}{rrr}  89 \\19 \\21.3483 \end{array}$ &
$\begin{array}{rrr}   187\\ 35 \\   18.7166 \end{array}$ &
$\begin{array}{rrr}   491\\ 59\\    12.0163   \end{array}$ &
$\begin{array}{rrr}   1881\\ 257 \\   13.6629   \end{array}$
\\
 \hline
 \end{tabular}
 \end{center}
\label{tab:klein}
\end{table}

\begin{figure}
\begin{center}
\includegraphics[width=0.7\textwidth]{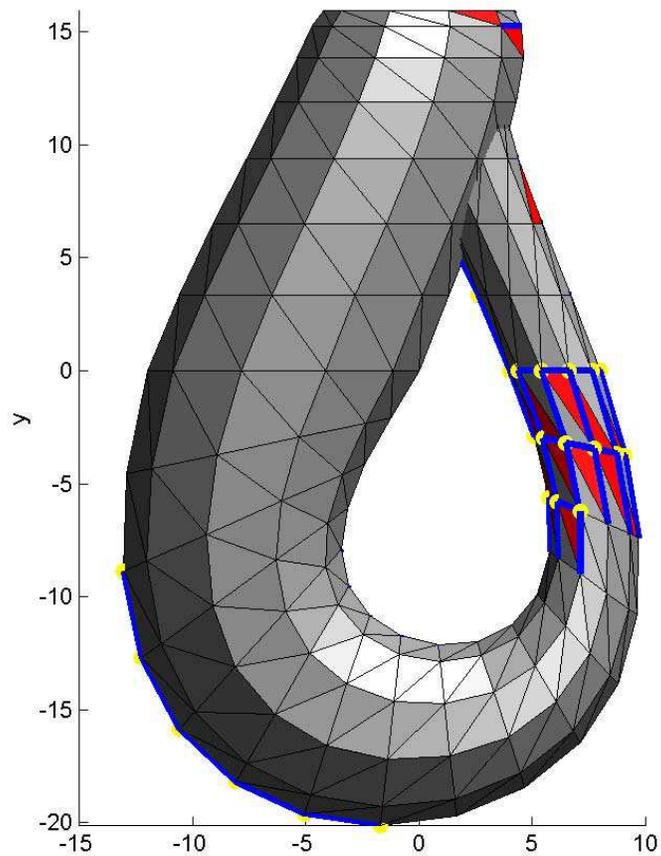}
  \caption{A triangulation of an (almost) Klein bottle immersed in $\R^3$ with 1881 simplices, with respect to a component-wise injective perturbation of the function defined of its vertices by $f(x,y,z)=(x,y)$, has 234 critical simplices: 73 critical vertices  out of 341  (in yellow), 128 critical edges  out of 940 (in blue),  and 56 critical faces  out of 600 (in red).
 }
\end{center}
\label{fig:klein}
\end{figure}

In conclusion, our experiments on synthetic data confirm that the current simplex-based matching algorithm scales well with the size of the complex.

\subsection{Examples on Real Data}
We consider four triangle meshes (available at \cite{sf}).
For each mesh the input  2-dimensional measuring function $f$
takes each vertex $v$ of coordinates $(x, y, z)$ to the pair $f(v) = (|x|, |y|)$.

In Table \ref{tab:space}, the first row shows on the top line the number of vertices in each considered mesh, and in
the middle line same quantities referred to the cell complex $\sC$
obtained by using our matching algorithm to reduce $\cK$. Finally, it also displays in the bottom line the ratio between
the second and the first lines, expressing them in percentage points.
The second and the third rows show similar information for the edges and the faces. Finally, the fourth row show the same
information for the total number of cells of each considered mesh $\cK$. The critical simplexes are displayed in Figure \ref{fig:space}.

Our experiment confirm that the current simplex-based matching algorithms produce a fair rate of reduction for simplices of any dimension also on real data. In particular, it shows a clear improvement with respect to the analogous result presented in \cite{AlKaLa17} and obtained using a vertex-based and recursive matching algorithm.

More experiments for a modified but equivalent version our algorithm can be found in ~\cite{Iur16}.

%\begin{landscape}
\begin{table}
\caption{Reduction performance on some natural triangle meshes.}
\begin{tabular}{| c | c | c | c | c |}
\hline
Dataset & \tt{tie} & \tt{space\_shuttle} & \tt{x\_wing} & \tt{space\_station}\\
 \hline
$\begin{array}{rrr}\# \cK_0\\ \# \sC_0\\ \% \end{array} $&
$\begin{array}{rrr}2014\\ 553\\   27.4578
 \end{array}$&
$\begin{array}{rrr}2376\\225\\  9.4697 \end{array}$ &
$\begin{array}{rrr}3099\\614\\ 19.8128 \end{array}$&
$\begin{array}{rrr}5749\\1773\\ 30.8401 \end{array}$ \\
 \hline
 $\begin{array}{rrr}\# \cK_1\\ \# \sC_1\\ \% \end{array}$ &
 $\begin{array}{rrr}5944\\1195\\  20.1043
 \end{array}$&
 $\begin{array}{rrr}6330\\243\\  3.8389 \end{array}$ &
 $\begin{array}{rrr}9190\\1232\\   13.4059 \end{array}$ &
 $\begin{array}{rrr}15949\\2556\\ 16.0261 \end{array}$ \\
 \hline
 $\begin{array}{rrr}\# \cK_2\\ \# \sC_2\\ \% \end{array}$ &
 $\begin{array}{rrr}3827\\539\\  14.0841 \end{array}$ &
 $\begin{array}{rrr}3952\\16\\  0.4049 \end{array}$&
 $\begin{array}{rrr}6076\\603\\   9.9243 \end{array}$ &
 $\begin{array}{rrr}10237\\820\\  8.0102 \end{array}$ \\
 \hline
  $\begin{array}{rrr}\# \cK\\ \# \sC\\ \% \end{array}$ &
 $\begin{array}{rrr}11785\\2287 \\ 19.4060  \end{array}$ &
 $\begin{array}{rrr}12658\\484 \\3.8237  \end{array}$ &
 $\begin{array}{rrr}18365\\  2449\\ 13.3351 \end{array}$ &
 $\begin{array}{rrr}31935\\5149 \\ 16.1234 \end{array}$\\
 \hline
 \end{tabular}
\label{tab:space}
\end{table}
%\end{landscape}

\begin{figure}
\begin{center}
\makebox[\textwidth]{\begin{tabular}{cc}
{\includegraphics[width=0.55\textwidth]{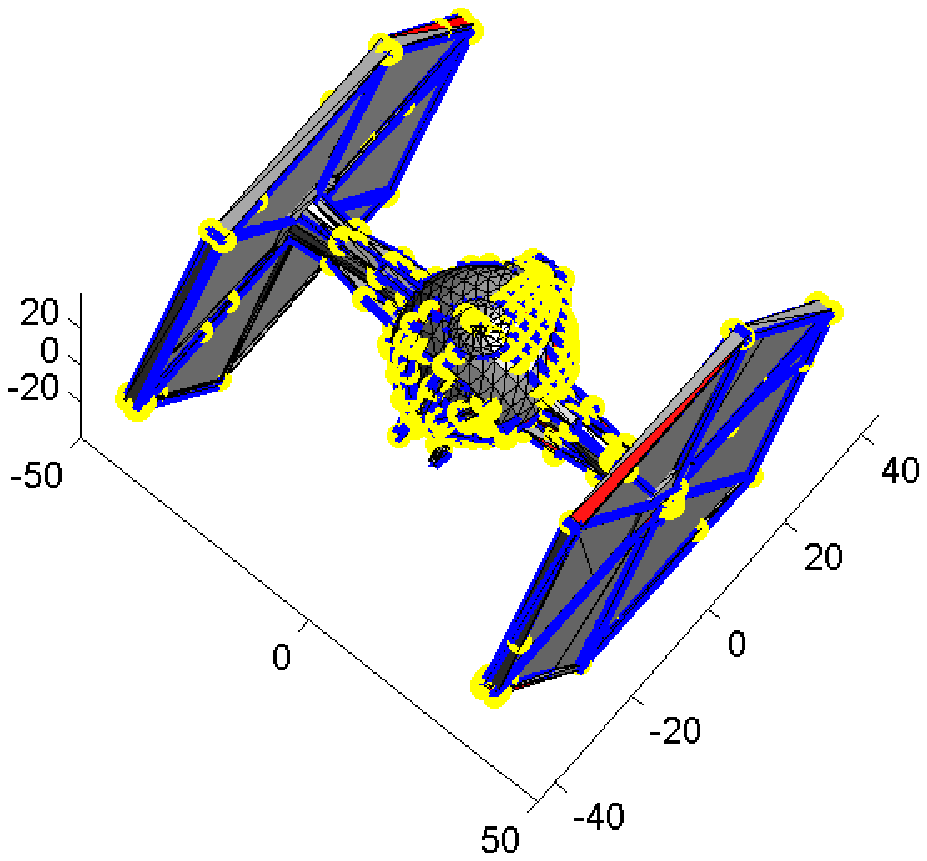}}&
{\includegraphics[width=0.55\textwidth]{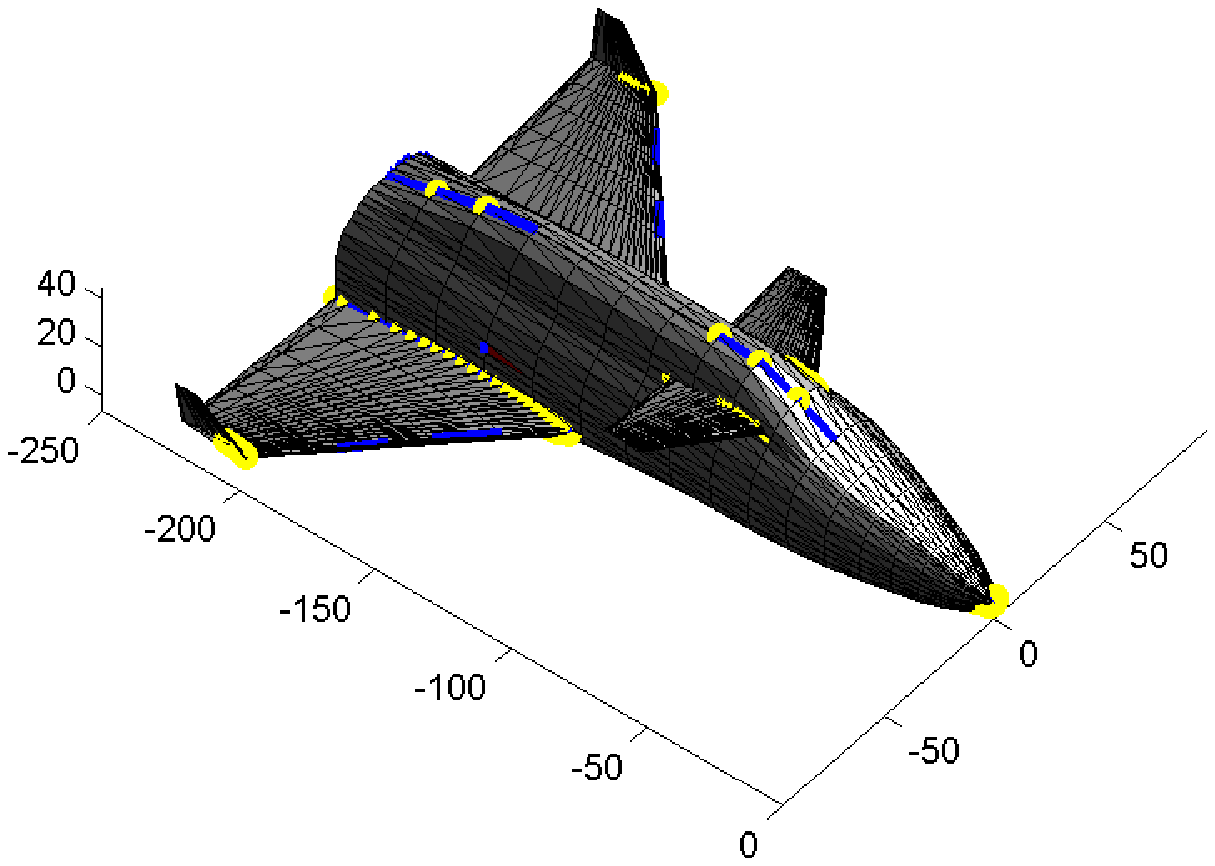} }\\
\tt{tie} & \tt{space\_shuttle} \\
{\includegraphics[width=0.55\textwidth]{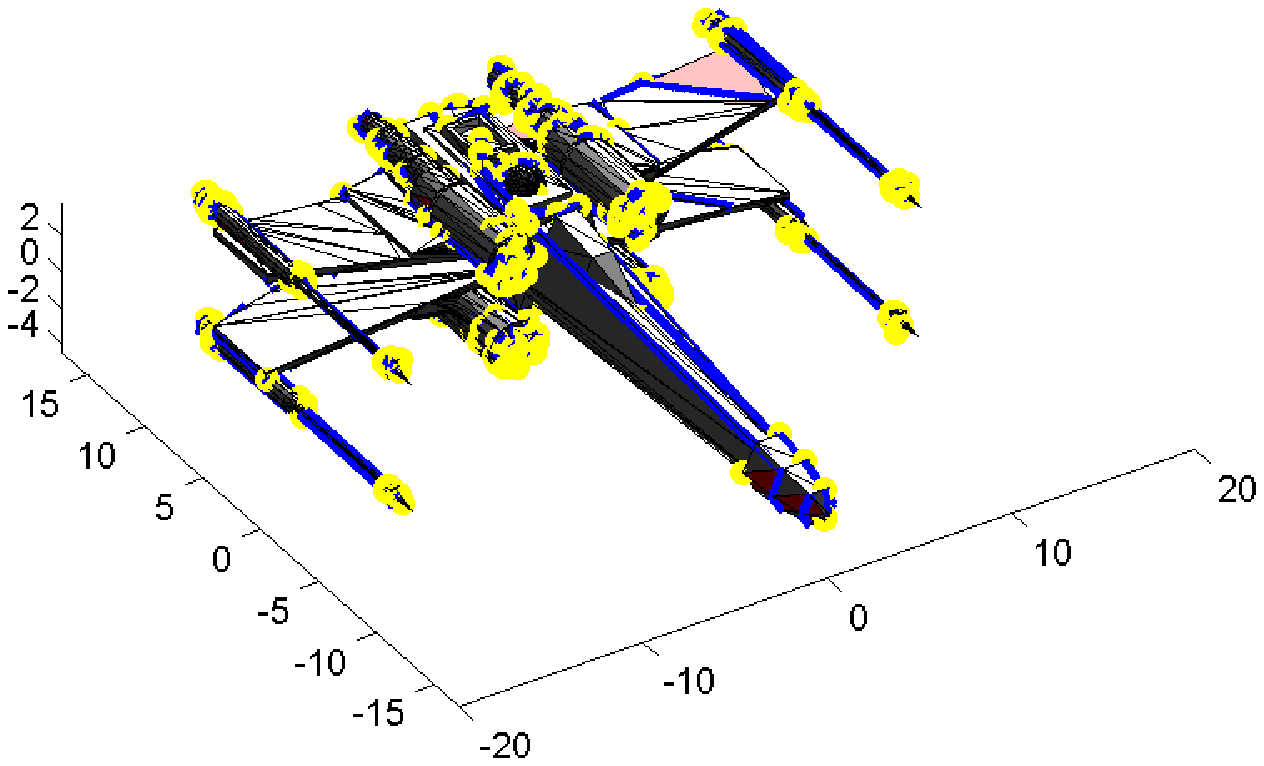}}&
{\includegraphics[width=0.55\textwidth]{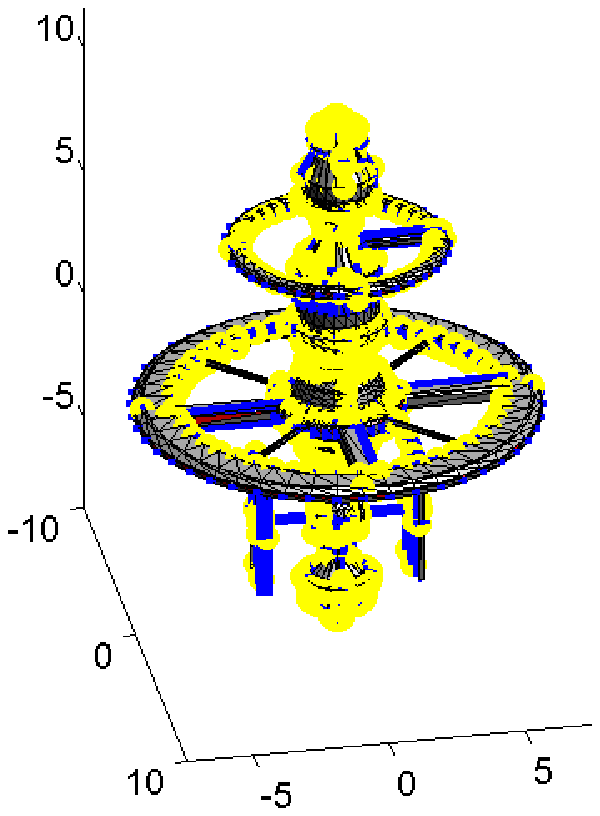}}\\
\tt{x\_wing} & \tt{space\_station}
\end{tabular}}
  \caption{The critical simplexes of four triangle meshes with respect to a component-wise injective perturbation of $f(x,y,z)=(|x|,|y|)$: critical vertices  in yellow,  critical edges  in blue,  and  critical faces  in red.
 }
\end{center}
\label{fig:space}
\end{figure}

\bibliographystyle{abbrv}
\bibliography{biblio}

\medskip

%\newpage

\noindent Department of Computer Science\\
Bishop's University\\
Lennoxville (Qu\'ebec),  Canada J1M 1Z7\\
mallili@ubishops.ca
\\~\\
D\'epartement de math\'ematiques\\
Universit\'e de Sherbrooke,\\
Sherbrooke (Qu\'ebec), Canada J1K 2R1\\
t.kaczynski@usherbrooke.ca
\\~\\
\noindent Dipartimento di Scienze e Metodi dell'Ingegneria\\
Universit\`a di Modena e Reggio Emilia\\
Reggio Emilia, Italy
\\
claudia.landi@unimore.it

\end{document}